\documentclass[letterpaper, 10 pt, conference]{ieeeconf}  

\IEEEoverridecommandlockouts                              

\usepackage[sort, nocompress]{cite}
\usepackage{amsmath,amssymb,amsthm, amsfonts, mathrsfs}
\usepackage{graphicx}
\usepackage{dsfont}
\usepackage{xcolor}
\usepackage{comment}
\usepackage{hyperref} 

\usepackage{cleveref}
\Crefname{equation}{}{}
\crefname{equation}{}{}
\usepackage{footmisc}
\usepackage{csquotes}
\let\labelindent\relax
\usepackage{enumitem}
\usepackage{booktabs}

\usepackage{algorithm}
\usepackage{algpseudocode}

\usepackage{caption, subcaption}

\def\BibTeX{{\rm B\kern-.05em{\sc i\kern-.025em b}\kern-.08em
    T\kern-.1667em\lower.7ex\hbox{E}\kern-.125emX}}

\theoremstyle{plain}
\newtheorem{theorem}{Theorem}

\theoremstyle{definition}

\newcommand{\ie}{\textit{i.e., }}
\newcommand{\eg}{\textit{e.g., }}

\newcommand{\cpt}{c} 
\newcommand{\sol}{F} 

\title{\LARGE \bf
Neural Spline Operators for Risk Quantification in Stochastic Systems
}

\author{Zhuoyuan Wang$^{1}$, Raffaele Romagnoli$^{2}$, Kamyar Azizzadenesheli$^{3}$ and Yorie Nakahira$^{1}$
\thanks{$^{1}$Zhuoyuan Wang and Yorie Nakahira are with the Department of Electrical and Computering Engineering,
        Carnegie Mellon University, Pittsburgh, USA.
        {\tt\small \{zhuoyuaw,ynakahir\}@andrew.cmu.edu}}%
\thanks{$^{2}$Raffaele Romagnoli is with the School of Science and Engineering, Department of Mathematics and Computer Science, Duquesne University, Pittsburgh, USA.
        {\tt\small romagnolir@duq.edu}}%
\thanks{$^{3}${\tt\small kaazizzad@gmail.com}}%
\thanks{This work is sponsored in part by the PRESTO Grant Number JPMJPR2136 from the Japan Science and Technology Agency, in part by the National Science Foundation under Grant No. 2442948, and in part by the Department of the Navy, Office of Naval Research, under award number N00014-23-1-2252. The views expressed are those of the authors and do not reflect the official policy or position of the US Navy, Department of Defense, or the US Government.}%
\thanks{The authors are grateful to Prof. Giovanni Leoni at Carnegie Mellon University for insightful discussions on the continuity of PDEs.}%
}

\begin{document}

\maketitle
\thispagestyle{empty}
\pagestyle{empty}

\begin{abstract}

Accurately quantifying long-term risk probabilities in diverse stochastic systems is essential for safety-critical control. However, existing sampling-based and partial differential equation (PDE)-based methods often struggle to handle complex varying dynamics. Physics-informed neural networks learn surrogate mappings for risk probabilities from varying system parameters of fixed and finite dimensions, yet can not account for functional variations in system dynamics. To address these challenges, we introduce physics-informed neural operator (PINO) methods to risk quantification problems, to learn mappings from varying \textit{functional} system dynamics to corresponding risk probabilities. Specifically, we propose Neural Spline Operators (NeSO), a PINO framework that leverages B-spline representations to improve training efficiency and achieve better initial and boundary condition enforcements, which are crucial for accurate risk quantification. We provide theoretical analysis demonstrating the universal approximation capability of NeSO. We also present two case studies, one with varying functional dynamics and another with high-dimensional multi-agent dynamics, to demonstrate the efficacy of NeSO and its significant online speed-up over existing methods. The proposed framework and the accompanying universal approximation theorem are expected to be beneficial for other control or PDE-related problems beyond risk quantification.







\end{abstract}

\section{Introduction}

Accurately quantifying long-term risk probabilities online across a wide range of stochastic systems is crucial to designing safe control policies~\cite{wang2022myopically}. Estimating these probabilities is non-trivial due to the rare nature of risk events and the long time horizons that need to be accounted for, resulting in heavy computational cost for sampling-based methods~\cite{rubino2009rare}. Partial differential equation (PDE)-based approaches alleviate this computational burden by directly characterizing the evolution of risk probabilities of different initial states and time horizons through PDEs, but they often struggle to scale efficiently online with varying forms of system dynamics~\cite{chern2021safe}. Physics-informed neural network (PINN) methods can quantify risk probabilities online for multiple systems by accepting parameters of system dynamics as input~\cite{wang2023generalizable, wang2025physics}. However, these approaches are inherently limited as they cannot accommodate inputs with varying parameter dimensions or system dynamics of different functional forms. In this paper, we aim to efficiently and accurately estimate long-term risk probabilities (or equivalently, safety probabilities) for diverse stochastic systems characterized by varying dynamic functions and safe region descriptions. We specifically address the limitations of existing PINN-based methods and introduce an approach capable of simultaneously handling systems with different functional dynamics.

Long-term safety probabilities in stochastic systems can often be described by solutions of certain convection-diffusion equations~\cite{chern2021safe}. Traditional numerical methods for solving these PDEs include grid-based and sampling-based approaches, both of which face significant challenges in accounting for varying functional forms of system dynamics. Recent PINN methods overcome some limitations by learning parameterized PDEs corresponding to safety probabilities of parameterized system dynamics. These networks are trained to approximate PDE solutions directly, enabling the prediction of safety probabilities across the entire state space and over various long-term time horizons. However, these PINNs remain restricted to fixed finite dimensional parameterized forms of system dynamics and cannot directly handle dynamics described by arbitrary functions.

To overcome these challenges, we introduce physics-informed operator learning methods~\cite{kovachki2023neural, li2024physics} into risk quantification, enabling direct learning of mappings from general functional system dynamics to risk probability values. While operator learning and neural operators significantly broadens applicability in many problems~\cite{azizzadenesheli2024neural}, direct application to risk quantification is insufficient due to difficulties in enforcing PDE constraints effectively, which is crucial for accuracy in safety-critical scenarios. Thus, we propose Neural Spline Operators (NeSO), a novel physics-informed neural operator (PINO) learning framework leveraging B-spline representations~\cite{ahlberg2016theory} to efficiently enforce PDE constraints during training. The spline-based representation inherently offers enhanced smoothness, inital and boundary condition enforcement, and in particular analytical derivative calculations, which is essential but previously computationally expensive for physics-informed learning~\cite{raissi2019physics, li2024physics}. We also theoretically demonstrate the universal approximation capability of our proposed NeSO framework. The primary contributions of this paper are summarized as follows:

\begin{itemize}
\item Introduction of physics-informed operator learning methods to risk quantification to enable simultaneous handling of varying functional system dynamics.
\item Development of the Neural Spline Operator (NeSO), employing a spline-based representation for efficient physics-informed operator learning, accompanied by theoretical guarantees on universal approximation.
\item Demonstration of effectiveness and significant computational speed-ups across various systems, including those with diverse functional dynamics and high-dimensional multi-agent settings.
\end{itemize}



\section{Problem Formulation}
\label{sec:problem_formulation}


We consider the following stochastic dynamical systems
\begin{equation}
\label{eq:dynamics}
    dx_t = f(x_t, t) dt + \sigma dW_t,
\end{equation}
where $x \in \mathcal{X} \in \mathbb{R}^n$ is the state, $W_t \in \mathbb{R}^n$ is the standard Wiener process with $W_0 = \mathbf{0}$, and $\sigma \in \mathbb{R}^n$ is the magnitude of the noise, and $f: \mathbb{R}^{n+1} \rightarrow \mathbb{R}^n$ is the (closed-loop) dynamics of the system. We define the safe set as
\begin{equation}
\label{eq:safe_set}
\mathcal{C}_\alpha =\left\{x \in \mathbb{R}^n: \phi_\alpha(x) \geq 0\right\},
\end{equation}
where $\phi_\alpha: \mathbb{R}^n \rightarrow \mathbb{R}$ is a function parameterized by $\alpha$. 
Specifically, in this paper we consider safe sets (or sub safe sets for multi-agent systems introduced in section~\ref{sec:exp_multi_agent}) with $\phi_\alpha = - \sum_{k=1}^n \mathds{1} (x_k \notin [\alpha_k^-, \alpha_k^+])$ where $\alpha_k^\pm \in \mathbb{R}$ for all $k = 1, \cdots, n$ and $\mathds{1}(\cdot)$ is the indicator function, \ie the safe region is characterized by the $n$-dimensional bounding box with edge values $\alpha_k^\pm$. This characterization of safe set is used in many scenarios such as lane keeping in autonomous driving~\cite{gangadhar2022adaptive}, dynamic walking of humanoid-robot~\cite{nguyen20163d}, \textit{etc}.
The risk quantification problem aims to estimate the following long-term safety probability under dynamics~\eqref{eq:dynamics}.
\begin{equation}
\label{eq:safe_prob}
\begin{aligned}
     F(x, t) := \mathbb{P}(x_\tau \in \mathcal{C}_\alpha, \forall \tau \in [0, t] \mid x_{0} = x).
\end{aligned}   
\end{equation}
Eq.~\eqref{eq:safe_prob} characterizes the probability of the system~\eqref{eq:dynamics} remaining safe under time horizon $t$ starting from initial state $x$. 
Note that the safety probability $F$ in~\eqref{eq:safe_prob} is also dependent on the system dynamics~\eqref{eq:dynamics} and safe set function~\eqref{eq:safe_set}, but we drop the notations for conciseness.
This long-term safety probability is critical for designing safe control for stochastic systems~\cite{wang2022myopically}, and the goal is to accurately estimate such value for all $x \in \mathcal{X}$ and different time horizon $t > 0$, with varying dynamics function $f$ and safe set parameter $\alpha$.

\section{Related Work}
\label{sec:related_work}

\subsection{Risk Quantification}




A standard method to estimate~\eqref{eq:safe_prob} is to use sampling-based methods such as Monte Carlo (MC), where the system dynamics~\eqref{eq:dynamics} is simulated starting from state $x$ for time horizon $t$ with multiple independent trajectories, and the safety probability is empirically estimated via the ratio of safe trajectories.
Multiple techniques such as importance sampling~\cite{janssen2013monte,cerou2012sequential}, subset simulation~\cite{au2001estimation}, multi-level Monte Carlo~\cite{giles2015multilevel} are further developed to reduce the sample complexity of MC, but such methods are intrinsically designed to quantify risk of single dynamics with fixed safe set. In this work, we propose methods to learn mappings of varying system dynamics and safe sets to the safety probabilities.

Another method to estimate the long-term safety probability is to solve a corresponding partial differential equation. Specifically, from~\cite{chern2021safe} we know that the safety probability~\eqref{eq:safe_prob} is the solution of the following convection diffusion equation
\begin{equation}
\label{eq:safe_prob_pde}
\frac{dF(x,t)}{dt} = f(x, t)\frac{dF(x,t)}{dx} + \frac{1}{2} \sigma^\top \sigma \operatorname{tr} \left[\frac{d^2 F(x,t)}{dx^2}\right],
\end{equation}
with the following initial and boundary conditions (ICBCs)
\begin{equation}
\label{eq:safe_prob_icbc} 
F(x,0) = 1, \forall x \in \mathcal{C}_\alpha, \quad F(x, t) = 0, \forall x \in \partial \mathcal{C}_\alpha, \forall t > 0.
\end{equation}
When the system dynamics is fixed and simple (\eg $f$ is a linear function), the PDE associated with the safety probability either has analytical solution or can be solved effectively with numerical solvers (\eg finite element method-based solvers~\cite{dhatt2012finite} and finite volume method-based solvers~\cite{eymard2000finite}). However, when the system dynamics is changing and is a general function, traditional solvers usually require special treatment to maintain stability~\cite{jha2023assessment, houston2020eliminating} and cannot solve the safety probability PDEs in efficient and generic manners.
Physics-informed learning methods leverage both data samples and PDE models, and are capable of estimating the safety probability mapping from fixed or parameterized dynamics function~\cite{wang2023generalizable}. However, how to effectively estimate the safety probability with varying and general functional system dynamics remains an open challenge. We propose the neural spline operator method to address this challenge, by solving~\eqref{eq:safe_prob_pde} with varying functional $f$ efficiently.

\subsection{Neural Operators and B-Splines}


Neural operators are generalizations of neural networks by learning mappings between infinite-dimensional function spaces rather than finite-dimensional Euclidean spaces~\cite{kovachki2023neural}. Specifically, a neural operator $G_\theta$, parameterized by $\theta$, maps an input function $h$ from one function space to another, formally written as:
\begin{equation}
G_\theta(h): \mathcal{H}(\mathbb{R}^m) \rightarrow \mathcal{H}(\mathbb{R}^l),
\end{equation}
where $\mathcal{H}(\mathbb{R}^k)$ denotes an appropriate function space defined on a given domain $\mathbb{R}^k$, for $k \in \mathbb{N}^+$. The key distinction between neural operators and classical neural networks lies in their internal operations. The standard linear transformation $W x + b$ (weighted sum of inputs $x$, with weights $W$ and biases $b$) in multilayer perceptrons in neural networks is replaced by a continuous integral operator of the form:
\begin{equation}
\int k(x, y) f(x) \, dx + b(y),
\end{equation}
where $f(x)$ is the input function, $k(x, y)$ is an integral kernel, and $b(y)$ is a bias function. This integral operator formulation allows neural operators to accept input and produce output as functions at arbitrary resolutions.
These properties make neural operators very competitive solutions to applications where input is of multiple resolutions and cannot be parameterized, such as weather forecast~\cite{pathak2022fourcastnet}, turbulence modeling~\cite{li2022fourier} and in particular learning surrogate maps for the solution operators of PDEs under the physics-informed neural operator learning setting~\cite{li2024physics, wang2021learning}.
In this work, we leverage physics-informed neural operators (PINOs) to learn solutions of PDEs corresponding to safety probabilities, and address risk quantification problems for systems with varying functional dynamics.
To the best of our knowledge, this is the first work to explore physics-informed operator learning in the context of risk quantification. 
In addition, although there are multiple existing architectures for physics-informed neural operator learning such as Fourier neural operators (FNO)~\cite{li2020fourier}, graph neural operators (GNO)~\cite{li2020neural}, spectral neural operators (SNO)~\cite{fanaskov2023spectral} and DeepONets~\cite{lu2019deeponet}, expensive derivative calculations through automatic differentiation or finite difference methods are often required~\cite{li2024physics}, and the trained models may not always comply with initial and boundary conditions (ICBCs) that the PDE needs to satisfy~\cite{brecht2023improving}.
We develop a specific neural operator architecture based on B-splines that provides analytical derivative calculations and direct enforcement of ICBC constraints, to improve training efficiency for PINOs in risk quantification problems. 




B-splines are piece-wise polynomial functions derived from slight adjustments of Bezier curves, aimed at obtaining polynomial curves that tie together smoothly~\cite{ahlberg2016theory}. 
B-splines are used in combination with finite element methods~\cite{jia2013reproducing} and through the variational dual formulation~\cite{sukumar2024variational} to solve PDEs.
B-splines are used as basis functions for surface reconstruction~\cite{iglesias2004functional}, nonlinear system modeling~\cite{wang2022modeling}, image segmentation~\cite{cho2021differentiable}, and controller design for dynamical systems~\cite{chen2004learning}. 
B-splines with weights learned from neural networks (NNs) are used to approximate fixed or parameterized ODEs~\cite{romagnoli2024building} and PDEs~\cite{doleglo2022deep, zhu2024best}, and specifically PDEs for safety probabilities~\cite{wang2025physics}, but these methods cannot account for functional inputs, thus not varying functional system dynamics in risk quantification.
We leverage the advantages of B-splines and incorporate it into a physics-informed neural operator learning framework, to learn mappings of safety probabilities from functional dynamics and varying safe sets, for efficient risk quantification of diverse stochastic systems.




\section{Proposed Method}
\label{sec:proposed_method}

In this section, we present our proposed neural spline operator (NeSO) method, which combines both data and PDE model for learning on functional spaces, to efficiently address the risk quantification problem with varying functional dynamics.
Specifically, we aim to design neural operators that learn a mapping from the dynamics function $f: \mathbb{R}^{n+1} \rightarrow \mathbb{R}^n$ and the safe set parameter $\alpha \in \mathbb{R}^\xi$, to the safety probability $F: \mathbb{R}^{n+1} \rightarrow \mathbb{R}$.

\subsection{B-Splines and Basis Functions}

In this section, we introduce B-splines as basis functions to approximate the safety probabilities.
We start by considering the one-dimensional cases.
For one-dimensional state $x \in \mathbb{R}$ (assuming fixed time $t$), the B-spline basis functions are given by the Cox-de Boor recursion formula:
\small
\begin{equation}
\label{eq:Cox_de_Boor}
B_{i,d}(x) = \frac{x - \hat{x}_i}{\hat{x}_{i+d} - \hat{x}_i} B_{i,d-1}(x) + \frac{\hat{x}_{i+d+1} - x}{\hat{x}_{i+d+1} - \hat{x}_{i+1}} B_{i+1,d-1}(x),
\end{equation}
\normalsize
and
\begin{equation}
\label{eq:Cox_de_Boor_interval}
B_{i,0}(x) = \begin{cases}
1, & \hat{x}_i \leq x < \hat{x}_{i+1}, \\
0, & \text{otherwise}.
\end{cases}
\end{equation}
Here, $B_{i,d}(x)$ denotes the value of the $i$-th B-spline basis of order $d$ evaluated at $x$, and $\hat x_i \in (\hat{x}_i)_{i=1}^{\ell+d+1}$ is a non-decreasing vector of knot points. Since a B-spline is a piece-wise polynomial function, the knot points determine in which polynomial the parameter $x$ belongs. 
While there are multiple ways of choosing knot points, we use $(\hat{x}_i)_{i=1}^{\ell+d+1}$ with $\hat{x}_1 = \hat{x}_2 = \cdots = \hat{x}_{d+1}$ and $\hat{x}_{\ell+1} = \hat{x}_{\ell+2} = \cdots = \hat{x}_{\ell+d+1}$, and for the remaining knot points we select equispaced values. 
Note that for the proposed framework we define the B-spline basis functions in a normalized domain $[0, 1]$, and the basis functions for any safe set domain of interest $[\alpha_k^-, \alpha_k^+]$ can be obtained through rescaling.
We then define the control points
\begin{equation}
\label{eq:control_pts_1d}
    \cpt := [\cpt_{1}, \cpt_{2}, \dots, \cpt_{{\ell}}],
\end{equation}
and the B-spline basis functions vector
\begin{equation}
\label{eq:bs_functions_1d}
    B_d(x) := [B_{1,d}(x),  B_{2,d}(x), \dots, B_{\ell,d}(x)]^\top.
\end{equation}
Then, we can approximate the safety probability $\sol(x, t)$ at a fixed time $t$ with 
\begin{equation} \label{eq:approx_sol_1d}
    \hat\sol(x, t)= \cpt B_d(x).
\end{equation}
Note that with our choice of knot points, we ensure the initial and final values of $\hat{\sol}(x)$ coincide with the initial and final control points $c_1$ and $c_\ell$. This property will be used later to directly impose initial conditions and Dirichlet boundary conditions with the neural spline operator.

More generally, for a $n$-dimensional state $x = [x_1, \cdots, x_n] \in \mathbb{R}^n$ and time $t \in \mathbb{R}^+$, we can generate B-spline basis functions based on the Cox-de Boor recursion formula~\eqref{eq:Cox_de_Boor} and~\eqref{eq:Cox_de_Boor_interval} along each dimension $x_k$ and $t$ with order $d_k$ for $x_1, \cdots, x_n$ and $d_t$ for $t$.
We use $B_{i_k,d_k}(x_k)$ to denote the B-spline basis of order $d_k$ for the $i_k$-th function of $x_k$, and $B_{i_t,d_t}(t)$ for $i_t$-th function of $t$ with order $d_t$. 

The $n+1$-dimensional control point tensor will be given by 
$C = [\cpt_{i_1, \cdots, i_n, i_t}]_{\ell_1 \times \cdots \times \ell_n \times l_t}$, 
where $i_k$ is the $k$-th state index of the control point, and $\ell_k$ is the number of control points along the $k$-th dimension. Similarly $i_t$ and $\ell_t$ are the index and the number of control points along $t$. The safety probability on the $n+1$-dimensional space can then be approximated with B-splines and control points via
\vspace{-0.3em}
\begin{equation}
\label{eq:bs_approx_n_dim}
\begin{aligned}
    \hat{\sol}(x_1, & \cdots,x_n, t) = \sum_{i_1=1}^{\ell_1}  
    \cdots \sum_{i_n=1}^{\ell_n}  \\ \quad \sum_{i_t=1}^{\ell_t}  & c_{i_1,\cdots,i_n, i_t}B_{i_1, d_1}(x_1)\cdots B_{i_n, d_n}(x_n) B_{i_t, d_t}(t).
\end{aligned}
\end{equation}
We will use $C \cdot B_{\ell,d}(x, t)$ to denote the weighted sum of the right hand side of~\eqref{eq:bs_approx_n_dim} when the context is clear for the rest of the paper.




\begin{figure*}[t]
    \centering
    \includegraphics[width=0.86\textwidth]{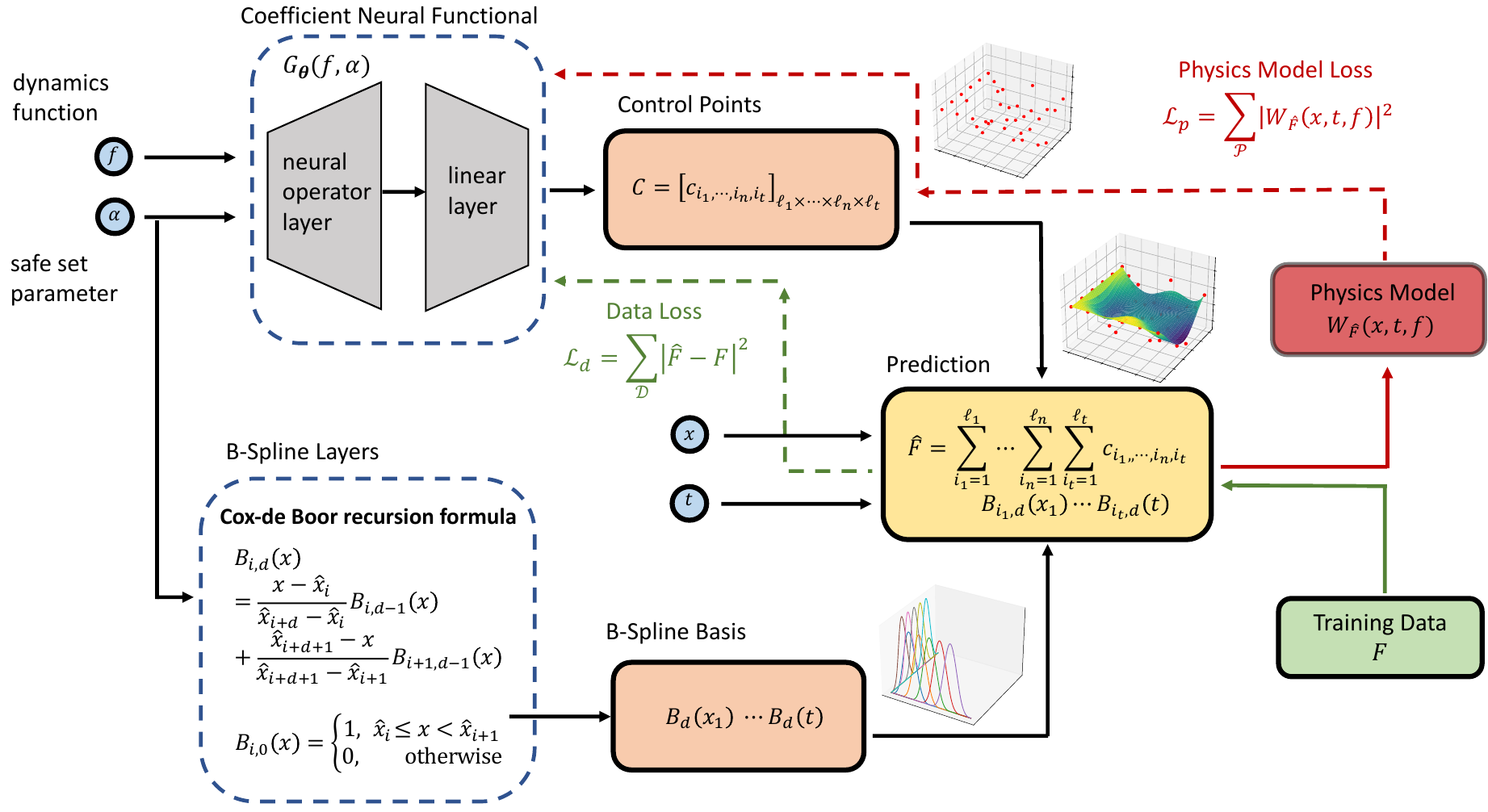}
    \vspace{0.9em}
    \caption{Diagram of the proposed neural spline operator (NeSO). The coefficient neural functional takes dynamics function $f$ and safe set parameters $\alpha$ as input and outputs the control points tensor $C$, which is then multiplied with the B-spline basis to produce the final output. Physics and data losses are imposed to train NeSO. Solid lines depict the forward pass, and dashed lines depict the backward pass.
    }
    \label{fig:neso_diagram}
\end{figure*}

\subsection{Neural Spline Operators}

In this section, we introduce our proposed neural spline operator (NeSO). The overall diagram of the framework is shown in Fig.~\ref{fig:neso_diagram}.
The framework composites a coefficient neural functional that learns the control point tensor $C$ given the dynamics function $f$ and parameters $\alpha$, and the B-spline layers that produce basis functions $B_{\ell, d}$. 

During the forward pass, the coefficient neural functional with parameter $\theta$, denoted by $G_{\theta}(f, \alpha)$, learns the control point tensor $C$ from the input dynamics function $f$ and parameters $\alpha$. The learned control points are then multiplied with the B-spline basis functions $B_{\ell, d}$ via~\eqref{eq:bs_approx_n_dim} to get the approximation $\hat \sol(x, t) =  G_{\theta}(f, \alpha) \cdot B_{\ell, d}(x,t)$, and we use $\hat \sol(x, t) = G_{\theta}(f, \alpha)(x, t)$ for simplicity. Note that inside the coefficient neural functional $C = G_{\theta}(f, \alpha)$, any initial condition or Dirichlet boundary condition of a PDE can be directly specified using the control point values. For example, in a 2D case where the initial condition is given by $\sol (x, 0) = 0, \forall x$, we can set the first column of the control point tensor $c_{i_1, 1} = 0$ for all $i_1 = 1, \cdots, \ell_1$ and this will ensure that the initial condition is met.


For the backward pass, two losses in PINO training are imposed to efficiently and effectively train the neural spline operator. 
We define the PDE residual
\begin{equation}
\begin{aligned}
    W_{\hat F}(x, t, f) := & \frac{d \hat F(x,t)}{dt} - f(x, t)\frac{d \hat F(x,t)}{dx}  \\
    & \qquad - \frac{1}{2} \sigma^\top \sigma \operatorname{tr} \left[\frac{d^2 \hat F(x,t)}{dx^2}\right].
\end{aligned}
\end{equation} 
We first impose a physics model loss 
\begin{equation}
\label{eq:physics_loss}
    \mathcal{L}_p = \frac{1}{|\mathcal{P}|} \sum_{x \in \mathcal{P}}  |W_{\hat F}(x, t, f)|^2,
\end{equation} 
where $\mathcal{P}$ is the set of points sampled to evaluate the governing PDE.
When data is available, we can additionally impose a data loss 
\begin{equation}
    \mathcal{L}_d = \frac{1}{|\mathcal{D}|}\sum_{(x, t) \in \mathcal{D}} |\sol(x, t) - \hat{\sol}(x, t)|^2,
\end{equation} 
to capture the mean square error of the approximation, where $\sol$ is the data point for the safety probability, $\mathcal{D}$ is the data set, and $\hat \sol$ is the prediction from the neural spline operator. 
The total loss is given by
\begin{equation}
    \mathcal{L} = w_p \mathcal{L}_p + w_d \mathcal{L}_d,
\end{equation}
where $w_p$ and $w_d$ are the weights for physics and data losses. Specific choices of $w_p$ and $w_d$ will be given in the case study section, and the general rule is to increase $w_p$ when the PDE model is accurate and increase $w_d$ when the data are reliable, and choose the weights so that the magnitude of the gradients matches the choice of the learning rate for training.

The coefficient neural functional composites a neural operator layer and a linear neural network layer, to achieve learning from functional spaces to control point tensors. The neural operator layer can take many architectures such as Fourier neural operators (FNO)~\cite{li2020fourier}, graph neural operator (GNO)~\cite{li2020neural}, U-shaped neural operator~\cite{li2020multipole,rahman2022u} and low-rank neural operator (LNO)~\cite{kovachki2023neural}. We treat the safe set parameters $\alpha$ as constant functions to adapt to existing neural operator architectures~\cite{yang2021seismic,shi2024broadband}. In the case studies in section~\ref{sec:experiments}, we specifically consider FNO as the neural operator layer, since it is widely accepted in the operator learning community. Due to space limits, we refer the readers to~\cite{li2020fourier} for the implementation details of FNO.


Note that there are several advantages of the proposed neural spline operator.
\begin{enumerate}
    \item The neural operator formulation enables mapping from \textbf{dynamics functionals} to safety probabilities, which covers a wide variety of system dynamics for risk quantification.
    
    \item The B-spline functions admit \textbf{analytical derivatives}~\cite{butterfield1976computation}, enables efficient training with physics-informed losses~\eqref{eq:physics_loss}. Specifically, the $p$-th derivative of the $d$-th ordered B-spline is given by
    \small
    \begin{equation}
    \label{eq:b-spline_derivatives}
    \begin{aligned}
        &\frac{d^p}{d x^p} B_{i, d}(x) = \\
        &\frac{(d-1)!}{(d-p-1)!} \sum_{k=0}^p(-1)^k\binom{p}{k} \frac{B_{i+k, d-p}(x)}{\prod_{j=0}^{p-1}\left(\hat x_{i+d-j-1}-\hat x_{i+k}\right)}.
    \end{aligned}
    \end{equation}
    \normalsize

    \item Control points of B-splines allow \textbf{direct assignment of ICBCs} without any learning involved~\cite{wang2025physics}, which improves accuracy and efficiency in learning PDEs for risk quantification.
    
\end{enumerate}

\section{Theoretical Analysis}
\label{sec:theoretical_guarantees}

In this section, we provide theoretical analysis on the universal approximation properties of the proposed neural spline operators.

We assume that the state space $\mathcal{X}$ is a compact set in $\mathbb{R}^n$, and the time horizon $t \in \mathcal{T}$ where $\mathcal{T}$ is a compact set in $\mathbb{R}$. Let $\Omega_1 := \mathcal{X} \subset \mathbb{R}^{n}$ be the space for the arguments of the input functions, and $\Omega_2 := \mathcal{X} \times \mathcal{T} \subset \mathbb{R}^{n} \times \mathbb{R}$ be the space for the arguments of the output functions. 
We assume the dynamics function $f \in \mathcal{K}_1 \subset C(\Omega_1; \mathbb{R}^n)$, where $C(\Omega_1; \mathbb{R}^n)$ is the space of continuous functions that map $\Omega_1$ to $\mathbb{R}^n$. 
Since the parameter $\alpha$ is input to the neural operator as a function, with slight abuse of notation, in this section, we denote $\alpha$ as a constant function on $\mathcal{X}$ and we assume $\alpha \in \mathcal{K}_2 \subset C(\Omega_1; \mathbb{R}^\xi)$.
Similarly, we assume the safety probability $F \in \mathcal{K}_3 \subset C(\Omega_2
;\mathbb{R})$.
We assume that $\mathcal{K}_1$, $\mathcal{K}_2$ and $\mathcal{K}_3$ are continuous Sobolev spaces~\cite{kovachki2021universal}, which are subsets of Banach spaces with norm $\|\cdot\|_C = \max_{\Omega}\|\cdot\|$.\footnote{With slight abuse of notation, here $C$ denotes the space of continuous functions instead of control points. Continuity of the Sobolev spaces can be guaranteed if $kp>n$, where $k$ is the order of the derivative, $p$ is integrability component, and $n$ is the dimension~\cite[Section 5.6.3]{evans2022partial}.} We further assume the mapping between $(f,\alpha) \rightarrow F$ is continuous.\footnote{Such continuity holds with regularity conditions on $f$ and $\mathcal{C}$~\cite{krylov1996lectures}.}



\begin{theorem}
\label{thm:neso_universal_approximation}
Given dynamical system~\eqref{eq:dynamics} and safe set~\eqref{eq:safe_set}, for any $\epsilon>0$,
there exists the number of control points $\ell_1, \cdots, \ell_n, \ell_t \in \mathbb{N}^+$, order of B-spline basis $d_1, \cdots, d_n, d_t \in \mathbb{N}^+$, and a parameter set $\theta$ for the coefficient neural functional $G$, such that
for any function $f \in \mathcal{F}$ and safe set parameter $\alpha \in \mathbb{R}^\xi$, and any $x \in \mathcal{X}$ and $t \in \mathcal{T}$, and the corresponding safety probability $F(x,t)$ defined in~\eqref{eq:safe_prob}, the following holds
\begin{equation}
\begin{aligned}
\big| F(x,t)-G_\theta(f, \alpha)(x, t) \big| \leq \epsilon,
\end{aligned}
\end{equation}
where
\begin{equation}
\label{eq:weighted_sum_neso}
    G_\theta(f, \alpha)(x, t) = G_\theta(f, \alpha) \cdot B_{\ell,d}(x, t),
\end{equation}
is the dot product of the control points and B-spline basis functions evaluated at $(x,t)$.

\end{theorem}

\begin{proof}
    See Appendix~\ref{sec:proof}.
\end{proof}

\section{Case Studies}
\label{sec:experiments}

In this section, we present two case studies for the proposed neural spline operator methods to estimate safety probabilities, specifically on a 1-dimensional system with varying functional dynamics in section~\ref{sec:exp_func_dyn}, and on a 14-dimensional multi-agent system with varying dynamics and safe sets in section~\ref{sec:exp_multi_agent}. Code is available at \href{https://github.com/jacobwang925/NeSO}{https://github.com/jacobwang925/NeSO}.

\subsection{Functional Dynamics}
\label{sec:exp_func_dyn}
We consider system~\eqref{eq:dynamics} with $x \in \mathbb{R}$ and safe set~\eqref{eq:safe_set} with $\alpha = 4$ and $\phi_\alpha(x) = x-4$ fixed. The dynamics functional of the system is given by
\begin{equation}
    f(x, t) = A_1 \sin\left( 2\pi \frac{\omega_1}{10} t + \psi_1 \right) + A_2 \sin\left( 2\pi \frac{\omega_2}{10} t + \psi_2 \right),
\end{equation}
where $A_1, A_2 \sim \mathcal{U}(-1, 1)$ are random amplitudes, $\omega_1, \omega_2 \sim \mathcal{U}(0.5, 2)$ are randomly chosen frequencies, and $\psi_1, \psi_2 \sim \mathcal{U}(0, 2\pi)$ are random phase shifts. Here, $\mathcal{U}(a, b)$ denotes uniform distribution in $[a, b]$. With probability 0.5, we set  $A_2 = 0$, reducing the function to a single sine wave.
The goal is to estimate the following long-term recovery probability, which is a slight variant of~\eqref{eq:safe_prob}.
\begin{equation}
\label{eq:recovery_prob_def}
    F(x,t) := \mathbb{P}\left( \exists \tau \in [0, t], \text{ s.t. } x_\tau \in \mathcal{C}_\alpha \mid x_{0} = x \right).
\end{equation}
This characterizes the probability of the system recovering to the safe region for at least once under time horizon $t$ starting from an initial state $x$. Such probability can be estimated either via MC, or through solving~\eqref{eq:safe_prob_pde}, with the following modified initial and boundary condition~\cite{chern2021safe}
\begin{equation}
\label{eq:recovery_icbc}
    F(x,0) = 0, \forall x \in \mathcal{C}_\alpha^c, \quad F(x, t) = 1, \forall x \in \partial \mathcal{C}_\alpha, \forall t > 0,
\end{equation}
where $\mathcal{C}_\alpha^c$ is the complement of $\mathcal{C}_\alpha$, \ie unsafe region. In this case study, our aim is to estimate~\eqref{eq:recovery_prob_def} in the state-temporal space $(x, t) \in [-10, 4] \times [0,10]$, with varying system dynamics $f$.

The ground truth recovery probability can be calculated via the cumulative integral of the hitting time density:
\begin{equation}
    F(x, t) = \int_{0}^{t} \frac{(4 - x)}{\sqrt{2 \pi \tau^3}} \exp\left(-\frac{((4 - x) - S(\tau))^2}{2 \tau} \right) d\tau,
\end{equation}
where the integrated effect of the dynamics function is
\begin{equation}
    S(t) = \int_{0}^{t} f(x_\tau, \tau) \, d\tau.
\end{equation}

For the choice of neural operators, we compare the proposed NeSO method with Fourier neural operator (FNO)~\cite{li2020fourier} as the backbone of the neural coefficient functional, against an architecture only of FNO, to learn the mappings from the dynamics function $f$ directly to the recovery probability $F$. In both methods, we use 8 Fourier modes in spatial and temporal dimensions with width of 32, and 3 spectral convolution layers in the FNO module.
We train both NeSO and FNO with 10 independent random functions and corresponding ground truth solutions for data loss function, and impose physics loss based on~\eqref{eq:safe_prob_pde}, with weights $w_d = 3$ and $w_p = 1$. Additional loss on the residual of initial condition and boundary condition (ICBC)~\eqref{eq:recovery_icbc} is imposed for FNO with weight $w_{\text{ICBC}} = 10$. We train both models with Adam optimizer with an initial learning rate of 0.001 for 500 epochs. Both methods converge during training, and Fig.~\ref{fig:1d_compare_visual} visualizes the prediction results on a random testing function $f$. It can be seen that both methods generalize to unseen dynamics to provide accurate estimation, due to the functional learning properties of neural operators. Table~\ref{tab:recovery_results} summarizes the training time and prediction errors on 10 independent random testing functions for both methods. Mean square errors (MSE), mean absolute errors (MAE) and average relative errors are reported. It can be seen that with the B-spline representation, the proposed NeSO outperforms FNO in both training efficiency and prediction accuracy, due to the analytical derivative calculations during training and the direct assignment of ICBCs.

\begin{figure}
    \centering
    \includegraphics[width=0.48\textwidth]{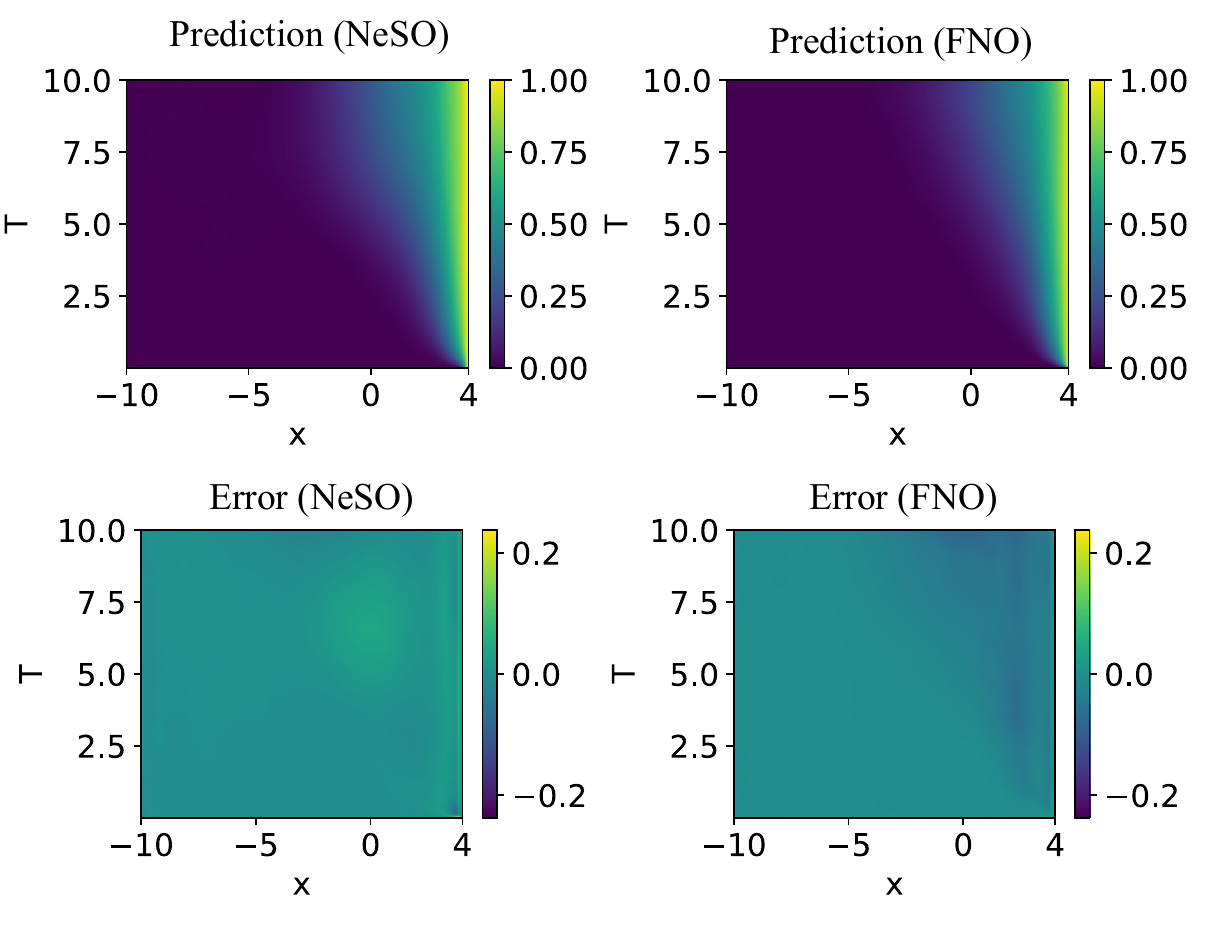}
    \vspace{-1.3em}
    \caption{Prediction results on a random testing function $f$ with NeSO and FNO.}
    \label{fig:1d_compare_visual}
\end{figure}

\begin{table}
    \centering
    \renewcommand{\arraystretch}{1.2}
    \setlength{\tabcolsep}{6pt}
    \caption{Test performance for recovery probability estimation over 10 test cases. Errors are reported as mean $\pm$ standard deviation.}
    \label{tab:recovery_results}
    \begin{tabular}{lcc}
        \toprule
        Metric & NeSO (Proposed) & FNO \\
        \midrule
        MSE ($\times 10^{-3}$) & $0.285 \pm 0.190$ & $5.298 \pm 9.296$ \\
        MAE ($\times 10^{-2}$) & $0.892 \pm 0.380$ & $2.827 \pm 2.968$ \\
        Rel. Error ($\times 10^{-2}$) & $5.697 \pm 1.888$ & $14.994 \pm 10.093$ \\
        Train Time (s) & $756$ & $1086$ \\
        \bottomrule
    \end{tabular}
    \vspace{-0.5em}
\end{table}

\subsection{High-Dimensional Multi-Agent Dynamics}
\label{sec:exp_multi_agent}


We consider a multi-agent system consisting of $N=7$ agents. Each agent $\Sigma_k$ is a mass-spring-damper system of the following dynamics,
\begin{equation}
\label{eq:multi_agent_sub_system}
d \, x^{(k)}_t = \left( A x^{(k)}_t + B u^{(k)}_t \right) dt+ \sigma_k w^{(k)}_t,
\end{equation}
where $x^{(k)} := [p^{(k)}, v^{(k)}] \in \mathbb{R}^2$ is the state, $u^{(k)} \in \mathbb{R}$ is the control input, and $w^{(k)} \in \mathbb{R}^2$ is the standard Wiener process to capture the noise, and $\sigma_k = 0.2 \; I_2 \in \mathbb{R}^2$ is the noise magnitude. Here, $I_N$ is the $N$-dimensional identity matrix and we will use this notation for the rest of this section. The system matrices are defined as $A = \begin{bmatrix} 0 & 1 \\ -\beta_1 & -\beta_2 \end{bmatrix}$ and $B = \begin{bmatrix} 0 \\ 1 \end{bmatrix}$, where $\beta_1, \beta_2 > 0$ are the physical parameters governing the spring and damping effects, respectively. The output of each agent is given by  
\begin{equation}
    y^{(k)}_t = C\,x^{(k)}_t, 
    \quad 
    \text{where } C = \begin{bmatrix} 1 & 0 \end{bmatrix}.
\end{equation}  
The control input $u^{(k)}_t$ of each dynamic agent is determined by the outputs $y^{(i)}_t$ of its neighboring agents, thereby enabling mutual interaction among the agents. Specifically, the control input is given by  
\begin{equation}
\label{eq:multi_agent_control}
u^{(k)}_t = -\sum_{i \in \mathcal{I}(k)} l_{k, i} \; y^{(i)}_t,
\end{equation}
where $\mathcal{I}(k) \subseteq\{1,2, \ldots, N\}$ denotes the set of neighboring agents of $\Sigma_k$, and $l_{k, i}$ is the feedback control gain that characterizes the interaction force of agent $i$ on agent $k$.


The agents interact through an undirected graph, encoded in a Laplacian matrix $L \in \mathbb{R}^{N \times N}$ whose $(i, k)$-th entry is defined by $L_{i, k}=l_{i, k}$, representing the interaction network among the agents. Specifically, in this case study we consider
\begin{equation}
L \;=\;
\begin{bmatrix}
5 & -1 & -1 & -1 & -1 & -1 & 0\\
-1 & 3 & 0 & -1 & 0 & 0 & -1\\
-1 & 0 & 2 & 0 & -1 & 0 & 0\\
-1 & -1 & 0 & 4 & -1 & -1 & 0\\
-1 & 0 & -1 & -1 & 4 & -1 & 0\\
-1 & 0 & 0 & -1 & -1 & 3 & 0\\
0 & -1 & 0 & 0 & 0 & 0 & 1
\end{bmatrix}.
\end{equation}
The overall system dynamics for the multi-agent system are obtained by stacking the states of all agents into the collective state vector $x_t = \bigl[x^{(1)\top}_t, \dots, x^{(N)\top}_t\bigr]^\top \in \mathbb{R}^{2N}$, yielding  
\begin{equation}
\label{eq:multi_agent_dynamics}
    dx_t = H\,x_t \, dt + \sigma w_t,
\end{equation}  
where $w_t \in \mathbb{R}^{2N}$ is the standard Wiener process and $\sigma = \text{diag}(\text{diag}(\sigma_1), \cdots, \text{diag} (\sigma_N)) = 0.2 \, I_{2_N}$. The matrix $H$ governs the coupled system dynamics and is defined as  
\begin{equation}
    H = I_N \otimes A - L \otimes (BC).
\end{equation}  
where $\otimes$ denotes the Kronecker product, and $BC$ describes the coupling between agents through their outputs.

The safe set of the multi-agent system is defined as~\eqref{eq:safe_set} with 
\begin{equation}
    \phi_\alpha(x) = \min_{k \in \{1, \dots, N\}} \left( \alpha_k - \left\|\bigl(\mathbf{t}_k^\top \otimes I_2\bigr)\,x \right\| \right),
\end{equation}  
where $\mathbf{t}_k$ is the $k$-th eigenvector of the Laplacian matrix $L$, and $\alpha_k \in \mathbb{R}^+$ is the corresponding safety threshold for agent $k$. The goal is to estimate the safety probability~\eqref{eq:safe_prob} for the entire multi-agent system~\eqref{eq:multi_agent_dynamics}.

Let $\lambda_k$ be the corresponding eigenvalue of the Laplacian matrix $L$ for mode $k$, let 
\begin{equation}
    \phi_{\alpha_k}(x^{(k)}) := \alpha_k - \left\|\bigl(\mathbf{t}_k^\top \otimes I_2\bigr)\,x \right\|,
\end{equation}
and we define safe sets for sub-systems as 
\begin{equation}
    \mathcal{C}_{k} := \{x^{(k)}: \phi_{\alpha_k}(x^{(k)}) \geq 0\}.
\end{equation}
The safety probability for sub-system $x^{(k)}$ is defined as
\begin{equation}
\label{eq:sub_sys_safe_function}
    F_k\bigl(x^{(k)},t\bigr):=\mathbb{P}(x^{(k)}_\tau \in \mathcal{C}_{k}, \forall \tau \in [0,t] \mid x^{(k)}_0 = x^{(k)}),
\end{equation} 
and we know its value evolves according to the following PDE~\cite{yasunaga2024orthogonal}
\begin{equation}
\label{eq:sub_sys_pde}
    \frac{\partial F_k}{\partial t} + v \,\frac{\partial F_k}{\partial p} - (\gamma_k\, p + \beta_2 \,v) \frac{\partial F_k}{\partial v} + \frac{\sigma_k^2}{2}\,\nabla^2 F_k = 0, 
\end{equation}  
with the following ICBC
\begin{equation}
    F_k\bigl(x^{(k)},0\bigr) = 1,  \forall x^{(k)} \in \mathcal{C}_{k},
    \;
    F_k\bigl(x^{(k)},t\bigr) = 0, \forall x^{(k)} \in \partial \mathcal{C}_{k}.
\end{equation}
Here, $ \gamma_k = \beta_1 + \lambda_k $, and the spatial domain is constrained by a threshold $[-\alpha_k, \alpha_k]$.
It is shown in~\cite[Theorem 1]{yasunaga2024orthogonal} that the overall safety probability can be calculated via  
\begin{equation}
\label{eq:multiplication}
    F(x,t) = \prod_{k=1}^N F_k\bigl(x^{(k)},t\bigr), 
\end{equation}
which is the multiplication of the safety probabilities for all sub-systems.

We consider fixed $\beta_1 = 1$, and varying $\beta_2 \sim \mathcal{U}(0.5, 2)$ and $\alpha_k \sim \mathcal{U}(1, 2)$ for different agent $k$. We then generate training data on 20 random systems with ground truth safety probability for each sub-system~\eqref{eq:sub_sys_safe_function} generated through solving~\eqref{eq:sub_sys_pde} with numerical solvers. With that, we train a NeSO that maps $(\beta_2, \lambda_k, \alpha_k)$ to $F_k$ for all $x^{(k)} \in [-\alpha_k, \alpha_k]$ and $t \in [0, 10]$. After training, we test on $\beta_2 = 1$ and $[\alpha_1, \cdots, \alpha_7] = [2, 2, 1, 1, 1, 1, 1]$, which is \textbf{not} seen during training. Fig.~\ref{fig:multi_safe_prob} shows the safety probability prediction results for the full 14-dimensional system, via prediction from NeSO on~\eqref{eq:sub_sys_safe_function} and multiplication in~\eqref{eq:multiplication}. The ground truth is obtained by running MC on the 14-dimensional multi-agent system~\eqref{eq:multi_agent_dynamics} with 10000 trajectories and empirically estimate~\eqref{eq:safe_prob}. It can be seen that NeSO reconstructs the safety probability of the high-dimensional multi-agent system with high accuracy, even though it did not train on data from this system. The time-averaged absolute difference between the MC and NeSO estimation is 0.0074, indicating the efficacy of the NeSO prediction. 
Table~\ref{tab:computation_times} shows the computation time for safety probability estimation of the 14-dimensional system through MC, solving PDEs~\eqref{eq:sub_sys_pde} for sub-systems, and the proposed NeSO. It can be seen that NeSO achieves a 2 orders of magnitude reduction in online computation time, compared to both the MC method and the PDE method. The training time for NeSO is 3999$\mathrm{s}$ on an L4 GPU. Based on this, we know that the total time for evaluating safety probabilities of $n$ dynamics and safe set instances is $4618n$ seconds for MC, $2079n$ seconds for the PDE method, and $3999 + 57n$ seconds for NeSO (including training). It can be seen that even with $n = 1$, the proposed NeSO method is faster than MC, and it will obtain better efficiency than the PDE method when the number of evaluated dynamics $n>2$.

\begin{figure}
    \centering
    \includegraphics[width=0.35\textwidth]{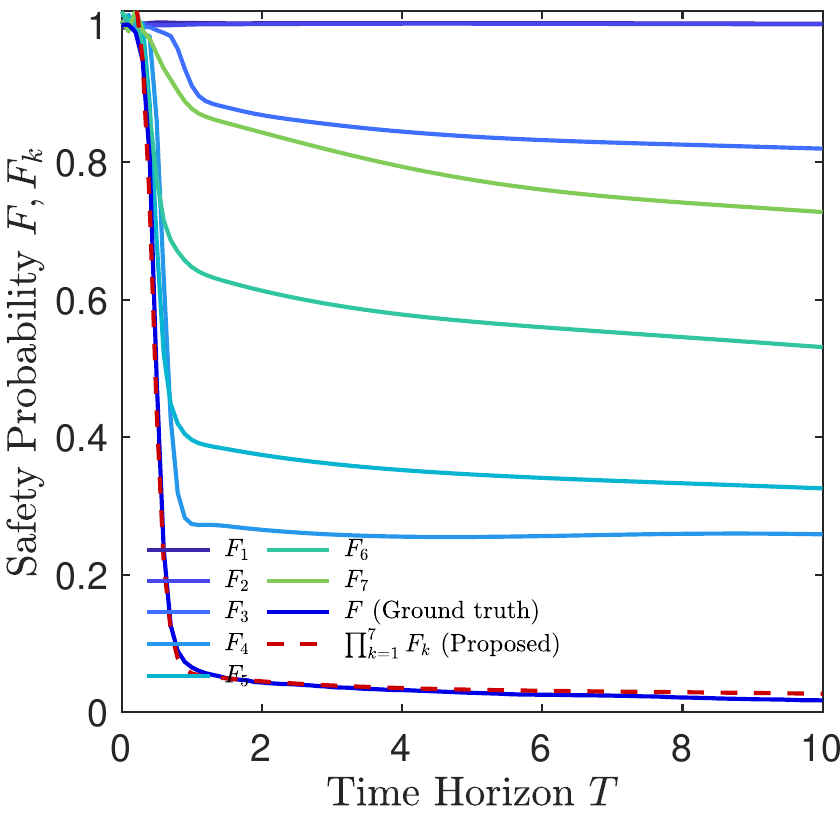}
    \caption{Safety probability estimation for the multi-agent system and its sub-systems.}
    \label{fig:multi_safe_prob}
\end{figure}

\begin{table}[t]
    \centering
    \renewcommand{\arraystretch}{1.2}
    \setlength{\tabcolsep}{8pt}
    \caption{Comparison of online computation times.}
    \label{tab:computation_times}
    \begin{tabular}{lccc}
        \toprule
        Method & MC & PDE & NeSO (proposed) \\
        \midrule
        Time (s) & 4618.48 & 2079.71 & \textbf{56.47}\\
        \bottomrule
    \end{tabular}
    \vspace{-0.5em}
\end{table}

\vspace{-0.3em}

\section{Conclusions}
\label{sec:conclusion}

In this paper, we consider the problem of risk quantification for stochastic systems with varying functional dynamics. We introduce physics-informed operator learning methods to this problem to directly find mappings from varying dynamic functionals and related parameters to the safety probabilities, while respecting all PDE constraints. Specifically, we propose neural spline operators, a framework for operator learning based on B-spline representations that improves training efficiency and estimation accuracy for risk quantification. We provide a theoretical analysis of the neural spline operator as universal approximators. We also show in experiments the efficacy of the proposed framework dealing with functional dynamics and high-dimensional multi-agent systems, where tremendous speedup is obtained compared to existing methods. 
The proposed framework and the accompanying universal approximation theorem are expected to benefit other related problems beyond risk quantification.


\bibliography{citation.bib}

\newpage

\appendix

\subsection{Proof of Theorem~\ref{thm:neso_universal_approximation}}
\label{sec:proof}

\begin{proof} (Theorem~\ref{thm:neso_universal_approximation})

We state the general proof of the universal approximation property of neural operators with basis functions, for learning continuous mappings between Sobolev function spaces. In the end, we show that this proof applies to our setting, where B-spline basis is considered and mappings to safety probabilities are studied.


\textbf{First, we prove the coefficient neural functional is a universal approximator of the optimal basis function weights.} We prove by construction. We start by showing the neural operator in the coefficient neural functional can approximate the target function, and show that a unique linear mapping can be learned by the linear layer in the coefficient neural functional to find the optimal weights.\footnote{Note that this is not the only possible setting for the coefficient neural functional to approximate the optimal basis function weights.}

We denote the input Sobolev function space as
\begin{equation}
\begin{aligned}
    \mathcal{H} &= H^k(\Omega_1; \mathbb{R}^{n+\xi}) := \Big\{ h \in L^2(\Omega_1; \mathbb{R}^{n+\xi}) \;\Big|\;\\
&\hspace{2cm} D^r h \in L^2(\Omega_1; \mathbb{R}^{n+\xi}),\; \forall |r| < k \Big\},
\end{aligned}
\end{equation}
where $k \in \mathbb{N}^+$ and $D^r h$ denotes the weak derivative of order $r$.
Similarly, we denote the output function space as
\begin{equation}
\begin{aligned}
    \mathcal{Y} &= Y^k(\Omega_2; \mathbb{R}) := \Big\{ g \in L^2(\Omega_2; \mathbb{R}) \;\Big|\;\\
&\hspace{2cm}D^r g \in L^2(\Omega_2; \mathbb{R}),\; \forall |r| < k \Big\}.
\end{aligned}
\end{equation}
Let $W : \mathcal{H} \to \mathcal{Y}$ be a continuous non-linear mapping, and let $h = (f, \alpha) \in \mathcal{K}_1 \times \mathcal{K}_2 \subset \mathcal{H}$. The goal is to approximate $g = W(h)$ via a neural operator, where $g\in \mathcal{K}_3 \subseteq \mathcal{Y}$ is the target function.

From the universal approximation theorem for neural operators (\eg FNO~\cite{kovachki2021universal}, DeepONet~\cite{lu2019deeponet}), there exists a continuous approximation of $W$ indicated as $\hat W$ defined over the same spaces, such that, for any $\epsilon_1>0$
\begin{equation}
\label{eq:no_layer_universal_approximation}
    \| W(h)-\hat W(h)\|<\epsilon_1.
\end{equation}
For a more compact notation, let us define $g_h\triangleq W(h)$ and define $\hat g_h\triangleq \hat W(h)$.

We then take $N < \infty$ orthogonal elements $\phi_i \in \mathcal{Y}$, and define a finite-dimensional subspace\footnote{The proof naturally generalizes to cases where $\phi_i \in \mathcal{Y}$ are not orthogonal, but \( B_{{\phi}_{i,j}} = \langle \phi_j, \phi_i \rangle \) is positive definite. In our case B-spline basis functions are orthogonal to each other.}
\begin{equation}
    V_N = \mathrm{span} \left\{ \phi_1, \ldots, \phi_N \right\} \subset \mathcal{Y}.
\end{equation}
For any function $g_h \in \mathcal{K}_3 \subseteq \mathcal{Y}$, we want to find a unique element $P_{V_N} g_h \in V_N$ such that
\begin{equation}
   \langle g_h - P_{V_N} g_h, v \rangle = 0 \quad \forall \; v \in V_N, 
\end{equation}
where $P_{V_N}$ is a projection from $\mathcal{Y}$ to $V_N$.
Since Sobolev spaces are Hilbert spaces, they are equipped with a norm $\Vert \cdot \Vert_{\mathcal{Y}}$ induced by the inner product. Considering
\begin{equation}
    P_{V_N} g_h = \sum_{j=1}^N c_j(g_h) \phi_j,
\end{equation}
we want
\begin{equation}
    \langle g - \sum_{j=1}^N c_j(g_h) \phi_j, \phi_i \rangle = 0 \quad \forall i = 1, \ldots, N.
\end{equation}
This equation can be written as
\begin{equation}
    \sum_{j=1}^N c_j(g_h) \langle \phi_j, \phi_i \rangle = \langle g, \phi_i \rangle.
\end{equation}
By defining the matrix \( B_{{\phi}_{i,j}} = \langle \phi_j, \phi_i \rangle \), which is invertible since $\phi_j$ are orthogonal thus linearly independent, we can find the coefficients of \( P_{V_N} \) denoted by $C^*$ with
\begin{equation}
    C^*(g_h) = B_{\phi}^{-1} b(g_h),
\end{equation}
where the elements of \( b(g_h) \) are \( b_i(g_h) = \langle g_h, \phi_i \rangle \).

For the uniqueness of the projection, consider the presence of \( v_1 \) and \( v_2 \) in \( V_N \) such that
\begin{equation}
    \langle g_h - v_1, v \rangle = \langle g_h - v_2, v \rangle = 0 \quad \forall v \in V_N.
\end{equation}
From the linearity of the inner product, we can rewrite the above equation as
\begin{equation}
    \langle v_1 - v_2, v \rangle = 0 \quad \forall v \in V_N.
\end{equation}
Since \( v_1 - v_2 \in V_N \), we can take \( v = v_1 - v_2 \), then we have
\begin{equation}
    \langle v_1 - v_2, v_1 - v_2 \rangle = 0.
\end{equation}
This implies that \( v_1 = v_2 \), which demonstrates the uniqueness of the projection.

For the vector \( b(g_h) \), we have
\begin{equation}
    \Vert b(g_h) \Vert_2^2 = \sum_{j=1}^N |\langle g_h, \phi_i \rangle|^2.
\end{equation}
By the Cauchy-Schwarz inequality
\begin{equation}
    |\langle g_h, \phi_i \rangle| \leq \Vert g_h \Vert_{\mathcal{Y}} \cdot \Vert \phi_i \Vert_{\mathcal{Y}},
\end{equation}
we get
\begin{equation}
    \Vert b(g_h) \Vert_2^2 = \sum_{j=1}^N |\langle g_h, \phi_i \rangle|^2 \leq \Vert g_h \Vert_{\mathcal{Y}}^2 \cdot \sum_{j=1}^N \Vert \phi_i \Vert_{\mathcal{Y}}^2.
\end{equation}
Hence,
\begin{equation}
    \Vert b(g_h) \Vert_2 \leq \left( \sum_{j=1}^N \Vert \phi_i \Vert_{\mathcal{Y}}^2 \right)^{1/2} \cdot \Vert g_h \Vert_{\mathcal{Y}}.
\end{equation}
Since \( B_\phi \in \mathbb{R}^{N \times N} \) is symmetric and positive-definite, the operator norm of \( B_{\phi}^{-1} \) is\begin{equation}
    \| B_{\phi}^{-1} \| = \frac{1}{\lambda_{\min}(B_{\phi})},
\end{equation}
where \( \lambda_{\min}(B_{\phi}) > 0 \) denotes the smallest eigenvalue of \( B_{\phi} \).
For any \( g_h, \hat g_h \in \mathcal{K}_3 \subseteq\mathcal{Y} \), we have:
\small
\begin{equation}
    \| C^*(g_h) - C^*(\hat g_h) \|_2 \leq \frac{1}{\lambda_{\min}(B_\phi)} \left( \sum_{i=1}^N \| \phi_i \|_{\mathcal{Y}}^2 \right)^{1/2} \| g_h - \hat g_h \|_{\mathcal{Y}}.
\end{equation}
\normalsize
We then define the following Lipschitz constants
\begin{equation}
    L_1 := \frac{1}{\lambda_{\min}(B_\phi)} \left( \sum_{i=1}^N \| \phi_i \|_{\mathcal{Y}}^2 \right)^{1/2},
\end{equation}
and
\begin{equation}
    L_2 := \lambda_{\max}(B_\phi),
\end{equation}
where $\lambda_{\max}(B_\phi)$ is the largest eigenvalue of $B_{\phi}$.
From the universal approximation of the neural operator layer~\eqref{eq:no_layer_universal_approximation}, we have
\begin{equation}
    \Vert g_h-\hat{g}_h \Vert_{\mathcal{Y}}\leq \epsilon_1,
\end{equation}
and we can choose $\epsilon_1=\frac{\epsilon}{4L_1 L_2}$ where $\epsilon>0$.
Hence, 
\begin{equation}
    \| C^*(g_h) - C^*(\hat g_h) \|_2 \leq \epsilon/(4L_2).
\end{equation}
Since \( C^*(\cdot) \) is a unique bounded linear mapping, it can be approximated by a neural network. 
By the universal approximation theorem~\cite{hornik1989multilayer}, there exists 
a neural network \( \hat{C}^* \) such that
\begin{equation}
    \|C^*(g_h) - \hat{C}^*(g_h)\|_2 < \epsilon_2.
\end{equation}
By using the triangular inequality, and choosing $\epsilon_2=\epsilon/(4L_2)$, we can write
\begin{equation}
\begin{aligned}
    & \|C^*(g_h) - \hat{C}^*(\hat{g}_h)\|_2 \\
    \leq \; & \|C^*(g_h) - C^*(\hat{g}_h)\|_2 +\|C^*(\hat{g}_h) - \hat{C}^*(\hat{g}_h)\|_2 \\
    \leq \; & \left(\frac{\epsilon}{4} + \frac{\epsilon}{4}\right)\cdot\frac{1}{L_2} = \frac{\epsilon}{2L_2}.
\end{aligned}
\end{equation}

\textbf{Then, we show that given the coefficient neural functional is a universal approximator of the optimal weights, and the basis functions has sufficient representation, the overall framework is a universal approximator.}

Let the basis \( \{ \phi_i \}_{i=1}^N \) and the number of basis functions \( N \) are chosen so that the subspace \( V_N \) is sufficiently rich to approximate any target function within a given tolerance \( \epsilon_3 > 0 \):
\begin{equation}
\label{bs_universal}
\| g_h-P_{V_N}g_h\|_\mathcal{Y}=\| g_h- B_{\phi}C^*(g_h)\|_\mathcal{Y}\leq \epsilon_3.
\end{equation}
Recalling that $g=W(h)$, the output of the operator $W(\cdot)$ with $h$ as input function, we can write
\begin{equation}
\begin{aligned}
    & \| g_h-B_{\phi}\hat{C}^*(\hat{g}_h) \|_\mathcal{Y} \\
    \leq \; & \| g_h- B_{\phi}C^*(g_h)\|_\mathcal{Y} + \|B_{\phi}\| \cdot \|(C^*(g_h)-\hat{C}^*(\hat{g}_h)\|_{2} \\
  \leq \; & \epsilon_3 +L_2\frac{\epsilon}{2L_2}=\epsilon_3+\frac{\epsilon}{2},
\end{aligned}
\end{equation}
By selecting $\epsilon_3=\epsilon/2$, we obtain
\begin{equation}
    \| g_h- B_{\phi}\hat{C}^*(\hat{g}_h)\|_\mathcal{Y}\leq \epsilon.
\end{equation}
In our case, we use B-splines as basis functions \eqref{eq:bs_approx_n_dim}, therefore
\begin{equation}
    B_{\phi}=B_{\ell,d}(x,t),
\end{equation}
and $C^*(g_h)$ represents the optimal control points tensor denoted by $C^*$. Since B-splines are universal approximators~\cite{wang2025physics}, there exists the number of control points $\ell_1, \cdots, \ell_n, \ell_t \in \mathbb{N}^+$, order of B-spline basis $d_1, \cdots, d_n, d_t \in \mathbb{N}^+$, such that 
\begin{equation}
    \Vert g_h - C^* \cdot B_{\ell,d}(x,t)\Vert_{\mathcal{Y}}\leq \epsilon_3.
\end{equation}
Since the coefficient neural functional produces $\hat C^*(\hat g_h)$ and is denoted by $G_{\theta}(h)=G_{\theta}(f,\alpha)$, we have
\[
\| g_h- G_{\theta}(f,\alpha)\cdot B_{\ell,d}(x,t)\|_\mathcal{Y}\leq \epsilon,
\]
and for our setting, $g_h = F$ is the safety probability.

\end{proof}

\end{document}